\RequirePackage{fix-cm}

\documentclass{svjour3}                     
\smartqed  
\usepackage{amssymb,amsmath, amsfonts }
\usepackage{xcolor}
\usepackage{graphicx}

\usepackage{tikz}
\definecolor{matlabcolor1}{rgb}{0, 0.4470, 0.7410}
\definecolor{matlabcolor2}{rgb}{0.8500, 0.3250, 0.0980}
\definecolor{matlabcolor3}{rgb}{0.9290, 0.6940, 0.1250}
\definecolor{matlabcolor4}{rgb}{0.4940, 0.1840, 0.5560}
\definecolor{matlabcolor5}{rgb}{0.4660, 0.6740, 0.1880}

\usepackage{enumitem}
\usepackage{lineno}
\begin{document}

\title{Impulsive fire disturbance in a savanna model:}

\subtitle{Tree-grass coexistence states, multiple stable system states, and resilience}


\author{Alanna Hoyer-Leitzel         \and 
        Sarah Iams 
}


\institute{Alanna Hoyer-Leitzel \at
              Department of Mathematics and Statistics, Mount Holyoke College, 50 College Street, South Hadley, MA 01075 \\
              \email{ahoyerle@mtholyoke.edu}           
           \and
           Sarah Iams \at
            John A. Paulson School of Engineering and Applied Sciences, Harvard University, Cambridge, MA 02138 \\
}
\date{Received: date / Accepted: date}

\maketitle

{\color{blue} }
{\color{red} }

\begin{abstract}
Savanna ecosystems are shaped by the frequency and intensity of regular fires. 
We model savannas via an ordinary differential equation (ODE) encoding a one-sided inhibitory Lotka-Volterra interaction between trees and grass. By applying fire as a discrete disturbance, we create an impulsive dynamical system that allows us to identify the impact of variation in fire frequency and intensity.
The model exhibits three different bistability regimes: between savanna and grassland; two savanna states; and savanna and woodland. 
The impulsive model reveals rich bifurcation structures in response to changes in fire intensity and frequency -- structures that are largely invisible to  analogous ODE models with continuous fire.
In addition, by using the amount of grass as an example of a socially-valued function of the system state, we examine the resilience of the social value to different disturbance regimes.  We find that large transitions (``tipping'') in the valued quantity can be triggered by small changes in disturbance regime.

\end{abstract}

\section{Introduction}
\label{intro}


Given the ecological, economic and cultural value of savannas as well as their precarious ecological role, savanna ecosystems are a frequent target of modeling investigations \cite{accatino2010tree,batllori2015minimal,baudena2010idealized,beckage2009vegetation,yatat2021minimalistic,goel2020dispersal,patterson2020probabilistic,ratajczak2017enemy,schertzer2015implications,staver2011tree,tamen2016tree,tamen2017minimalistic,touboul2018complex,wuyts2019tropical,yatat2018spatially}.  In tree-grass-fire interaction models of savanna ecosystems, the impact of fire on tree and grass biomass is often represented as a continuous mortality in ordinary differential equation (ODE) models \cite{accatino2010tree,beckage2009vegetation,yatat2021minimalistic,staver2011tree,touboul2018complex} and in spatial partial differential equations (PDE) models \cite{goel2020dispersal,wuyts2019tropical,yatat2018spatially}.  Continuous fire is an obvious simplification and in this paper we show that it is not equivalent to similar models with discrete fire.
Fire is also sometimes represented via a more complicated stochastically applied loss term \cite{batllori2015minimal,baudena2010idealized,patterson2020probabilistic}.  Stochastic loss terms lead to some models that are difficult to fully analyze. However, \cite{patterson2020probabilistic} shows that stochastic loss terms can be used to underpin continuum models, giving a nuanced relationship between the two types of models. 

In this paper, we use an impulsive modeling approach, where fire is imposed deterministically as a discrete disturbance on the state of the system.  This is similar to a model presented in \cite{tamen2016tree}.  This choice combines benefits (and shortcomings) of continuous and stochastic fire models.  The deterministic nature of fire in the impulsive model 
makes the model more easily amenable to analysis.  The assumption of periodic fire remains a simplification.  The discrete nature of disturbances in the impulsive model and in some stochastic models mimics the difference in time scales between fire and tree growth, and even between fire and grass growth.  Moreover, the discrete structure describes a fire disturbance more fully, making it possible to separate the impacts of fire timing and fire intensity.  

An impulsive differential equation has three parts: (1) a continuous differential equation that governs the system except for at a sequence of discrete impulses. (2) the impulse function describing the jumps or kicks in the state of the system; and (3) impulse mechanisms that describe when the impulses happen, either at certain times, at certain states or both \cite{roup2003limit}. The impulsive differential equation in this paper involves a Lotka-Volterra system with interacting grass and tree populations and a periodic, state-dependent impulse mechanism to represent fires.  In a non-impulsive continuous system, long-term behavior may approach a fixed amount of trees and grass.  The system might also exhibit oscillations.  In an impulsive system, the long-term behavior is time periodic whenever the impulse is non-zero.  For example, in this model, we see the impulse of fire mortality, and then slower recovery of the grass and trees.  It is possible to think of the long-term behavior in the non-impulsive system as capturing a time-average that smooths out the rapid changes driven by impulses in the impulsive system, or vice versa.  For example, Tamen et al \cite{tamen2016tree} work with a dimensionalized version of the model in this paper, and they show that parameters of the impulsive model can be adjusted so that the time average of the impulsive results are a good representation of the mean-fields for coexistence of grass and trees in a savanna in three different regions of Cameroon.

Using a time periodic impulse makes it possible to analyse the impulsive differential equation as a map from one post-kick state to the next post-kick state.  We refer to this map as a \textit{flow-kick map} \cite{meyer2018quantifying}.  We refer to the set of parameters describing the impulse as a \textit{disturbance regime}.  Figure \ref{fig:ImpulsiveDefinition} shows two time series for a one-dimensional impulsive differential equation. We see the flow (solid lines) and proportional removal of the state (impulse/kick, dashed lines). A fixed point of the flow-kick map (red disk in Figure \ref{fig:ImpulsiveDefinition}) corresponds to a periodic solution of the impulsive differential equation (red line segment in Figure \ref{fig:ImpulsiveDefinition}). Time series and phase planes for a two-dimensional impulsive system are shown in Figure \ref{fig:timeseriesphaseportrait}.

 Observations suggest that savanna systems are vulnerable to woody encroachment \cite{archer2017woody,stevens2017savanna}, potentially shifting the system from a savanna to a forest state, and reducing its value as rangeland for livestock.  Studies on bistability between savanna and forest \cite{staver2011tree,staver2011global} suggest that fire is an important mechanism in determining which state dominates.  This raises concerns about the resilience of the savanna ecosystem to differing fire disturbance regimes.    
Naively, resilience of a valued state of a system to a change in parameters is linked to the bifurcation structure of the system \cite{meyer2016mathematical,meyer2018quantifying}.  However, the changes in long term behavior that are identified via a bifurcation analysis may not directly relate to aspects of the system that are targets for preservation or management \cite{zeeman2018resilience}.  A function of the system state that is of social value may be a more appropriate focus for analysis.  For example, the forage value of grass in a savanna could be a conservation target.  Such a \emph{socially valued function} may respond to the changes in system state that occur at bifurcations, or may indefinitely remain within a socially valued region even as the system undergoes a bifurcation.
Bifurcation analyses typically focus on long term asymptotic limits of the system behavior.  We also include shorter term transient behaviors in this resilience analysis, potentially making it more relevant for decision support \cite{morozov2020long}.



 In this paper, we introduce (section \ref{sec:modelmethods}) and analyze (section 
 \ref{sec:results}) an impulsive savanna model, then identify the resilience of high-grass states (section \ref{sec:discussion}).  
  The model exhibits three bistability regimes: savanna and grassland; high and low tree density savanna; savanna and woodland (section \ref{sec:results}, see figure \ref{fig:four tau k1 stab} and table \ref{table:stability_regions}).  
  We use numerical methods (section \ref{sec:methods-numerical}) to generate some bifurcation results (section \ref{sec:numbifn}) and use analytical techniques to extend the stability analysis of the system in the disturbance parameter space. Section \ref{sec:BifThms} contains analytical proofs which correct some erroneous results in \cite{tamen2016tree}.
 We compare the bifurcation structure of the impulsive system to related continuous systems (section \ref{sec:ctsdisturb}) to show the intrinsic limitations of the disturbance regimes in a continuous system.
Using the amount of grass as a socially-valued function of the system state, we examine the resilience of that social value to varying disturbance (section \ref{sec:resanalysis}).  At some locations in the parameter space, small changes in the disturbance pattern lead to irreversible changes in the valued function.  This could be considered ``tipping'' (an irreversible change in state), and corresponds to woody encroachment into the savanna. 
The mathematical results presented in section \ref{sec:results} focus on long term predictions of the model. By considering how these results apply on shorter timescales we make the analysis relevant to timescales of human management (section \ref{sec:timescale}). Section \ref{sec:context} contextualizes the results of this model and section \ref{sec:future} offers future directions. This work shows there is value in using deterministic discrete fire disturbance models for insight into savanna mechanisms.

\begin{figure}
  \includegraphics{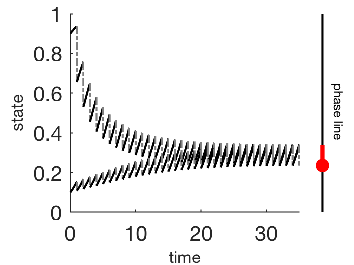}
\caption{The plot shows two time series, with different initial conditions, for an impulsive differential equation.  
In this example, the state variable flows according to $\dot x = \frac{1}{2}x(1-x)$ for a period of one time unit.  The state variable then experiences an impulse (kick) proportional to its state at the end of the flow period, $\Delta  x(t_n) = x(t_n^+) - x(t_n^-)=-0.3x(t_n^-)$ with $t_n = n$. The two time series are converging to an asymptotically stable periodic solution of the impulsive system, which is shown in red on the phase line.  
The red dot denotes the fixed point of the associated flow-kick map. }

\label{fig:ImpulsiveDefinition}       
\end{figure}

\section{Model and Methods}
\label{sec:modelmethods}

In this section, we introduce a Lotka-Volterra system along with an impulsive fire disturbance.  This system is closely related to the impulsive system introduced and analyzed in \cite{tamen2016tree,tamen2017minimalistic}.
We also present the numerical methods we use to identify fixed points and bifurcations of the impulsive system.

\subsection{Impulsive Lotka-Volterra Model}
\label{sec:model_intro}

\subsubsection{Underlying continuous tree-grass dynamics}

We consider a model, given in Equation \ref{eqn:system}, based on a competitive Lotka-Volterra system for trees, $T$, and grass, $G$, in a savanna. Competition is one-sided, with competition between trees and grass only affecting grass.  $G$ represents the biomass of grass, $T$ the biomass of trees, and $t'$ is time, measured in years.
\begin{equation}
    \begin{split}
    \label{eqn:system}
        \frac{dG}{dt'}&=\gamma_GG (1-\frac{G}{K_G}) - \gamma_{TG} GT\\
\frac{dT}{dt'}&=\gamma_T T(1-\frac{T}{K_T})
    \end{split}
\end{equation}
$K_G, K_T$ are the carrying capacities of grass and trees respectively, $\gamma_G, \gamma_T$ are growth rates for grass and trees respectively, and $\gamma_{TG}$, which we usually assume to be positive, is the parameter associated with the one-sided inhibitory interaction on grass by trees. The nondimensionalized version of this system is given in Equation \ref{eqn:nondim}.
\begin{equation}
    \begin{split}
    \label{eqn:nondim}
        \frac{dx}{dt}&=f(x,y) = x(1-x-\alpha y)\\
           \frac{dy}{dt}&=g(x,y) = \delta y(1-y)
    \end{split}
\end{equation}
In the nondimensional system, $x=G/{K_G}$, $y=T/{K_T}$, $t=\gamma_Gt'$, $\alpha = (K_T\gamma_{TG})/{\gamma_G}$, and $\delta = \gamma_T/{\gamma_G}$ 

We call equation \ref{eqn:nondim} the \emph{underlying system}.  This one-sided Lotka-Volterra system has four fixed points.  The fixed points and their stability are shown in Table \ref{tab:fp}. We use $\alpha = 0.8$, $\delta = 0.6$ through the analysis in this paper. These choices of $\alpha$ and $\delta$ correspond to the parameter set in Figure 5C of \cite{tamen2016tree}.

\begin{table}
\begin{center}
\begin{tabular}{| l | l | l |}
\hline
Fixed point & Description & Stability \\
\hline\hline
 $(0,0)$ & desert state & unstable for $\delta > 0$ \\
 \hline
  $(1,0)$ & grass-only state & saddle for $\delta > 0$ \\
  \hline
  $(0,1)$ & forest-only state & saddle for $\delta > 0, 0 < \alpha < 1$ \\
  & & stable for $\delta > 0$, $1 < \alpha$ \\
  \hline
    $(1-\alpha,1)$ & forest-grass state & stable for $\delta > 0, 0 < \alpha < 1$ \\
  & & saddle for $\delta > 0$, $1 < \alpha$ \\
  \hline
\end{tabular}
\caption{Fixed points and their stability for the underlying Lotka-Volterra system, given in Equation \ref{eqn:nondim}. This system never exhibits bistability.  For $\alpha = 0.8, \delta = 0.6$ (the parameter set used in this paper) the forest-grass coexistence state is the only stable fixed point of the underlying system. See Table \ref{table:stability_regions} for the wider range of stable and bi-stable states that exist in the impulsive system. 
\label{tab:fp}}
\end{center}
\end{table}


\subsubsection{Impulsive fire disturbances}

To encode fire in the system, we use impulses, referred to as \emph{kicks}.  These instantaneously adjust the state of the system, representing the impact of fire.  In the dimensional system, the impulse occurs with a period of $\tau'$ years, which is equivalent to a period of $\tau=\gamma_G\tau'$ in the nondimensional system.

The impact of fire on the grass state is modeled as a proportional kick, so that for each fire, some proportion, $k_1$, of the grass dies.  The mathematical statement encoding this impulse is given in Equation \ref{eqn:grass}.  
\begin{equation} 
\label{eqn:grass}
\Delta x(n\tau)=x(n\tau^+)-x(n\tau^-)= -k_1x(n\tau^-)
\end{equation}
For the tree state, fire mortality is modeled as both proportional loss of the pre-fire tree biomass, with proportion $k_2$, and as dependent on the amount of grass that burns.  The maximum tree mortality would be $k_2 y(n\tau^-)$, but it is modulated by a ``switching function", $\omega(k_1 x(n\tau^-))$, that specifies a proportion of the maximum mortality that will occur.  The mathematical statement encoding this impulse is given in Equation \ref{eqn:trees}.
\begin{equation}
\label{eqn:trees}
\Delta y(n\tau)=y(n\tau^+)-y(n\tau^-)=-k_2\ \omega(k_1x(n\tau^-))y(n\tau^-)
\end{equation}
We use 
\begin{equation}
    \label{eqn:omega}
    \omega(\xi) = \dfrac{\xi^2}{a^2+\xi^2},
\end{equation} where the parameter $a$ acts as a threshold setting the transition from input values associated with $\omega < \frac{1}{2}$ to input values associated with $\omega > \frac{1}{2}$. 
This means that when the amount of grass biomass burned, $k_1 x$, is less than $a$, the tree mortality is less than half its maximum, so $\vert\Delta y(n\tau)\vert < \frac{1}{2} k_2 y(n\tau^-)$, while when $k_1 x > a$, the tree mortality is more than half its maximum, $\vert\Delta y(n\tau)\vert > \frac{1}{2}k_2 y(n\tau^-) $. We use $a=0.08$ for the analysis in this paper. Such a switching function is motivated by percolation models for fire spread (e.g. \cite{schertzer2015implications}) and by empirical studies that show fire incidence has a threshold dependence on flammable cover (e.g. \cite{archibald2009limits}).

An example of this system is given in Figure \ref{fig:timeseriesphaseportrait}. Subfigures A and B show time series for the grass and trees from two different initial conditions. Subfigures C and D show the phase portrait; Subfigure D also includes the basins of attraction for the different stable periodic equilibria, and a ``slow region" to which solutions converge quickly, and then move along more slowly. This slow region is determined by the stable and unstable manifolds of the stable and saddle points, respectively. The method for determining the basins of attraction and slow region are given in Section \ref{sec:methods-numerical}.

\begin{figure}[h]
\centering
\includegraphics[]{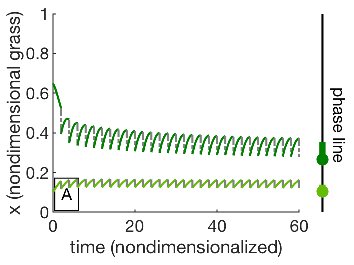}
\includegraphics[]{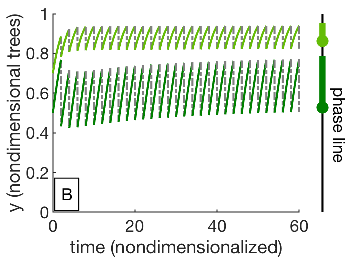}
\includegraphics[]{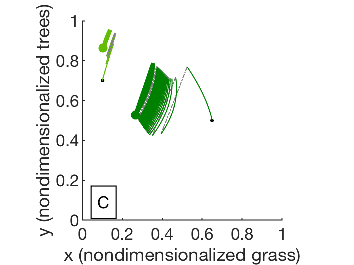}
\includegraphics[]{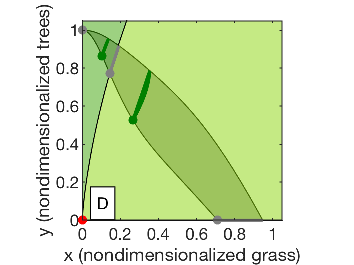}
\includegraphics[trim = .1in 1.5in .1in .55in, clip]{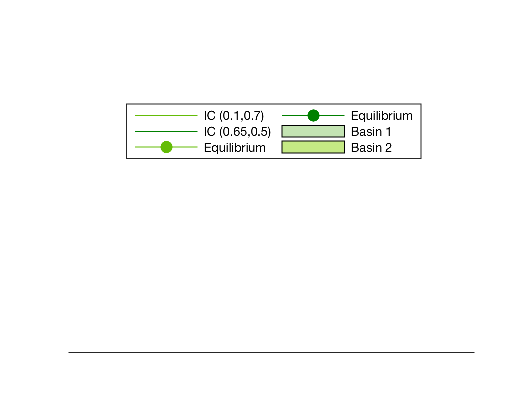}
\caption{\textbf{(A)} shows the time series for $x$ (grass) for two initial conditions of the impulsive differential equation. The long term flow-kick equilibria are marked on a phase line to the right.  The forest-dominant state is shown in light green and the grass-dominant in dark green. \textbf{(B)} shows the corresponding time series for $y$ (trees). \textbf{(C)} shows the parameterization of these two time series in phase space, again for 30 iterations of the flow-kick map. The dots denote fixed points of the flow-kick map, the curve along with the dot shows the periodic equilibrium of the impulsive system, and the straight line shows the action of the kick. In phase space, we observe rapid convergence to the unstable manifold of a saddle point, and slower movement along that manifold. \textbf{(D)} shows the basins of attraction and the six equilibria in the system. There are three saddle points: a forest-only state ($(0,1)$), a periodic grass-only state, and a savanna state. There are two stable attracting fixed points (both are tree-grass coexistence states). The region associated with the unstable manifold of the central saddle point is shaded a slightly darker color. Solutions converge quickly to this region and more slowly evolve towards their respective attracting fixed point. See \ref{sec:methods-numerical} for details on how the basins of attraction where calculated. 
The disturbance parameter set used to generate this figure is $k_1=0.25$, $k_2=0.6$, and $\tau=2$. }
\label{fig:timeseriesphaseportrait}   
\end{figure}

\subsection{Numerical Methods}
\label{sec:methods-numerical} 
We simulate the kick-flow map in Matlab using \texttt{ode45} to numerically integrate the differential equations from Equation \ref{eqn:nondim}, and then apply the kicks from Equations \ref{eqn:grass} and \ref{eqn:trees}. Stable fixed points can be found numerically by iterating the kick-flow map. We use a maximum of 300 iterations or convergence of the map to within $10^{-6}$ of the previous step.

We can also simulate the system using MatContM. Within MatContM, we encode the flow-kick map described in sections \ref{sec:model_intro}, rather than encoding a closed form version of flow-kick map (equation \ref{eqn:flowkickmap}). To compute fixed point continuations, identify bifurcations, and continue bifurcation numerically, MatContM \cite{meijer2017matcontm} requires Jacobian, Hessian, and third derivative information of the map. We generate these numerically, by running the flow-kick map from perturbed initial conditions.

Fixed point continuations in MatContM are particularly useful for finding saddle and unstable fixed points that we show in the phase portraits in Figure \ref{fig:timeseriesphaseportrait} and in the Appendix.  After finding a stable fixed point, we can continue the fixed point through bifurcations to find the saddle and unstable fixed point curves. The variable stepsize used by MatcontM means that the fixed points for an exact parameter set might not be included in the curve, so we use a newton solver to find a fixed point specifically chosen parameter values. For a saddle point, we then use the linearization of the map (calculated numerically) to identify points near the stable and unstable manifolds. We iterate points on the unstable manifold forward in time and the stable manifold backwards in time to find the boundaries of the slow region.

\section{Results}
\label{sec:results}

In this section, we examine the bifurcation and stability structure of the model. We look at this structure in the $\tau k_1 k_2$-disturbance space, considering only changes to fire intensity and frequency, and not changes to the underlying system. As shown in \cite{tamen2016tree}, for all values of $k_1$, $k_2$, and $\tau$, two fixed points always exist at $(0,0)$ and at $(0,1)$.  At $(0,0)$, with both vegetation types extinct, the fixed point is unstable (saddle or repellor).  At $(0,1)$, the fixed point corresponds to a  forest-only or woodland state, where grass is extinct and trees are at their carrying capacity. Numerically, we find there are up to four other fixed points in the system, depending on the disturbance regime.  See section \ref{sec:appendix} for detailed phase portraits of possible cases. 

In changing the disturbance regime, we observe various bistabilities, one of which was described in \cite{tamen2016tree}, and others that are novel.  These bistabilities are not present in the underlying continuous system, or in the continuous analog to the impulsive system (see Sections \ref{sec:ctsdisturb} and \ref{sec:ctsdistrubdiscuss}). The bistabilities mimic ecological bistability observations \cite{staver2011global,staver2011tree}. In the model, there are a total of seven qualitatively different stability and bistability regimes (see figure \ref{fig:four tau k1 stab}).  These different combinations of stable grass-only, forest-only, and savanna coexistence states and are shown in Table \ref{table:stability_regions}.

Using MatContM, we numerically identify the locations of transcritical and limit point (saddle-node or fold) bifurcations.  These bifurcations curves separate different bistability regions.  The curves are described in Section \ref{sec:numbifn} and shown in Figures \ref{fig:bifthms}, \ref{fig:four tau k1 stab}.  In Section \ref{sec:BifThms}, we use an algebraic condition that must be met when two fixed point branches cross to identify analytic expressions for the transcritical bifurcation curves that had been identified numerically in Section \ref{sec:numbifn}.

 \subsection{Numerically identified bifurcations}
\label{sec:numbifn}

\begin{figure}
\centering
\includegraphics[]{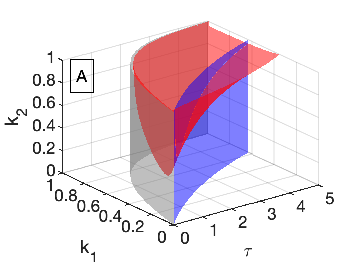}
\includegraphics[]{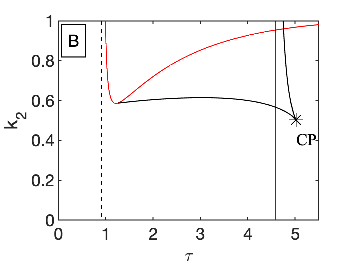}
\includegraphics[]{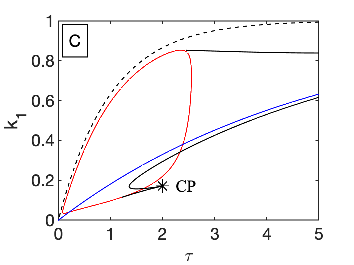}
\includegraphics[]{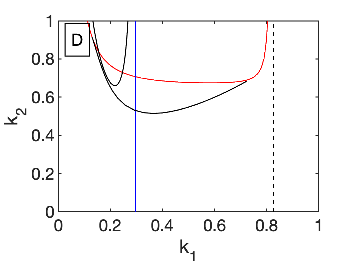}

\includegraphics[trim = .1in 1.25in 0in .5in ,clip]{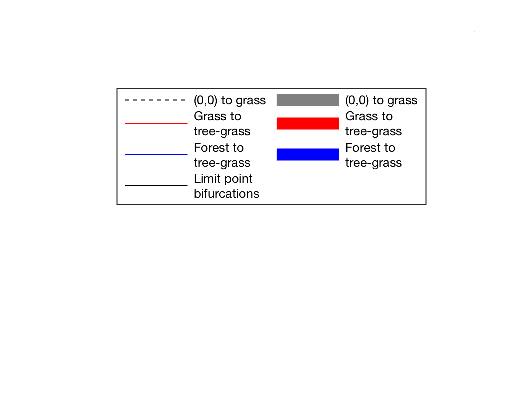}

\caption{ 
\textbf{(A)} shows bifurcation surfaces in disturbance space ($k_1,k_2,\tau$-space) for three analytically derived bifurcations. The transcritical bifurcation where a nontrivial grass-only fixed point, a grassland state, comes into existence as a non-attracting state (equation \ref{eqn:grass-fp existence}), is shown in gray.  The red surface shows the transcritical bifurcation between a grass-only fixed point and a tree-grass coexistence fixed point (equation \ref{eqn:thmgrassbp}), which can be thought of as a savanna state. In blue is the transcritical bifurcation between a forest-only fixed point and a tree-grass coexistence fixed point (equation \ref{eqn:thmforestbp}). Note that both the grass-only bifurcation and forest to savanna bifurcation points are not dependent on $k_2$. Figures B,C,D show cross-sections of these surfaces, along with numerically determined limit point bifurcations for these values of the parameters. \textbf{(B)} Cross-section when $k_1=0.6$. A co-dimension two bifurcation point, specifically a cusp point, is marked with a asterisk and labeled CP. \textbf{(C)} Cross-section when $k_2=0.8$. Again, a cusp point appears in the cross-section. \textbf{(D)} Cross-section when $\tau=1.75$.  The limit point and transcritical bifurcation curves divide disturbance space into a large number of regions, where each region has qualitatively distinct dynamics.  See Table \ref{table:stability_regions} for additional information.
}
\label{fig:bifthms}
\end{figure}

Using the methods described in Section \ref{sec:methods-numerical}, we find two types of bifurcations: transcritical bifurcations and limit point bifurcations.  There are three distinct transcritical bifurcations in the systems, involving the origin, a forest-only fixed point, and a grass-only fixed point.  Their locations are challenging to identify numerically for small values of $\tau$.  Analytic conditions for their locations are given in Theorems \ref{thm:grassbp}, \ref{thm:forestbp}, and Theorem 3.1 of \cite{tamen2016tree}.  

Figures \ref{fig:bifthms}BCD and \ref{fig:four tau k1 stab} show transcritical bifurcation curves as a dashed black line, a blue line, an a red line, respectively.  The transcritical bifurcation depicted by the dashed line is where a grass-only fixed point enters the first quadrant as a saddle.  This is the transition from region I to region II of Table \ref{table:stability_regions}.  The blue curves show the location of a transcritical bifurcation at $(0,1)$, where the forest-only state changes from stable to saddle and a coexistence state enters the first quadrant as a stable fixed point.  This is the transition from region IV to region VII or region V to region VI of Table \ref{table:stability_regions}. The red curves show the location of a transcritical bifurcation at $(x_g,0)$ where the grass-only state changes from saddle to stable, and a coexistence state enters the first quadrant as unstable fixed point.  Regions IV, VII, VIII of Table \ref{table:stability_regions} are inside the closed red loop made by cross-sections of this bifurcation surface.

In addition, we identify limit points curves.  The locations of these curves from numerical continuation are shown in black (solid lines) in Figures \ref{fig:bifthms}BCD and \ref{fig:four tau k1 stab}.  The fixed points associated with the limit point curve contradict the claim of a unique savanna tree-grass coexistence fixed point in \cite{tamen2016tree}. Regions V, VI, and VII each have more than one savanna tree-grass coexistence fixed point.

\subsection{Bifurcation Theorems} \label{sec:BifThms}

In Theorems \ref{thm:grassbp} and \ref{thm:forestbp} we find analytical expressions for a necessary condition for transcritical bifurcations at the grass-only and forest-only fixed points. These conditions are related to conditions presented in \cite{tamen2016tree}, match our numerical results for the bifurcations, and make it possible to identify the locations of bifurcations close to $k_1=\tau=0$.  The methods of proof are also based on ideas in \cite{tamen2016tree}.

In part of the $k_1k_2\tau$-parameter space,  $(0,0)$ and $(0,1)$ are the only fixed points of the system in the invariant first quadrant $\mathbb{R}_{\geq 0} \times \mathbb{R}_{\geq 0}$. This is region I in Figure \ref{fig:four tau k1 stab} and Table \ref{table:stability_regions}.  For parameter sets that meet a grass-existence condition, there also exists a fixed point of the flow kick map at \begin{equation}(x_g,0) = \left(\frac{1-(1-k_1)e^{\tau}}{1-e^{\tau}},0\right).\label{eqn:grass-only fixed point}\end{equation}  This fixed point crosses into positive $x$ through a transcritical bifurcation when \begin{equation} k_1=1-e^{-\tau}\label{eqn:grass-fp existence}\end{equation} (see Theorem 3.1 of \cite{tamen2016tree}). When $k_1 \leq 1-e^{-\tau}$,  $x_g$ is in $\mathbb{R}_{\geq 0} \times \mathbb{R}_{\geq 0}$, giving at least three fixed points of the system.  This condition is shown as a gray surface in Figure \ref{fig:bifthms}1 and as a dashed black line in Figures \ref{fig:bifthms}BCD.

The following two theorems give necessary conditions for tree-grass savanna coexistence fixed points where both $x$ and $y$ are nonzero.

\begin{theorem} \label{thm:grassbp}

A necessary condition for a transcritical bifurcation associated with a grass-tree coexistence fixed point crossing into the first quadrant, is met when \begin{equation} 1-e^{\delta\tau}\left(1-k_2\omega (\xi_g)\right)=0, \text{ with } \xi_g = \dfrac{k_1x_g}{1-k_1} \text{ for } k_1<1-e^{-\tau},\label{eqn:thmgrassbp}\end{equation}
at the point $(x_g^*,0)$ where $(x_g^*,0)$ is of the form in Equation \ref{eqn:grass-only fixed point} and $\omega(.)$ is defined in Equation \ref{eqn:omega}.
\end{theorem}

\begin{proof}
Let $\Phi(x,y,k_1,k_2,\tau)$ be the flow kick map such that
\begin{equation}
\label{eqn:flowkickmap}
\begin{bmatrix} x_{n+1} \\ y_{n+1}\end{bmatrix}=\Phi(x_n,y_n)=\begin{bmatrix} U(x_n,y_n)\\ V(x_n,y_n)\end{bmatrix}.\end{equation}

We assume that $k_1<1-e^{-\tau}$.  From equation \ref{eqn:grass-only fixed point}, a grass-only fixed point $(x_g,0)$ is defined by
\begin{equation}
\left\{\begin{array}{l}
    x_g=U(x_g,0)=\dfrac{1-(1-k_1)e^{\tau}}{1-e^{\tau}}\\
    0=V(x_g,0)\end{array}\right.
\end{equation}



A transcritical bifurcation occurs when the grass-only $(x_g,0)$ branch crosses a branch of fixed points of the form $(x^*,y^*)$ with $y^*$ not uniformly zero.  In other words, near the transcritical bifurcation, $y - V(x,y) = 0$ has a $y=0$ root, corresponding to the $(x_g,0)$ branch, and a root with $y$ almost always non-zero, corresponding to the $(x^*,y^*)$ branch.  The transcritical bifurcation occurs when these two branches cross so that the $y=0$ root of $y-V(x_g,y)=0$ has multiplicity two.  We construct a function $h(y)$ that reduces the multiplicity of the $y=0$ solution.

Let $h(y)=\dfrac{y-V(x_g,y)}{y}=1-\dfrac{V(x_g,y)}{y}$=0.
We give explicit formulas for $V(x,y)$ and $h(y)$. By integrating equation \ref{eqn:nondim}, and incorporating the impulses from equations \ref{eqn:grass} and \ref{eqn:trees}, we find
\begin{equation}
 V(x,y)=\left(1-k_2\omega(\xi)\right)\frac{e^{\delta\tau}y}{1+\delta y(-1+e^{\delta\tau})} \text{ where }\xi = k_1\frac{x}{1-k_1}.
\end{equation}
Given $x_g=\dfrac{1-(1-k_1)e^{\tau}}{1-e^{\tau}}$, $h(y)$ becomes
\begin{equation}
\label{eqn:hy}
h(y)=1-\dfrac{V(x_g,y)}{y}=1- \frac{e^{\delta\tau}\left(1-k_2\omega(\xi_g)\right)}{1+\delta y(-1+e^{\delta\tau})} \text{ where } \xi_g = k_1\dfrac{1-(1-k_1)e^\tau}{(1-e^\tau)(1-k_1)}.
\end{equation}

Setting Equation \ref{eqn:hy} to zero, we find a root of $h(y) = 0$  when $y = \dfrac{e^{\delta\tau}\left(1-k_2\omega(\xi_g)\right)-1}{\delta(-1+e^{\delta\tau})}$.  This crosses $y=0$ when $e^{\delta\tau}\left(1-k_2\omega(\xi_g)\right)-1=0$, giving a repeated root $y=0$ to the equation $y=V(x_g,y)$. Thus $e^{\delta\tau}\left(1-k_2\omega(\xi_g)\right)-1=0$ is a necessary condition for a transcritical bifurcation. \hfill$\square$

\end{proof}

The bifurcation condition in Theorem \ref{thm:grassbp}, conditioned on the grass-only fixed point being in $\mathbb{R}_{\geq 0} \times \mathbb{R}_{\geq 0}$, is shown as as a red surface in Figure \ref{fig:bifthms}A and as a red curve in Figures \ref{fig:bifthms}BCD.  The condition in Theorem \ref{thm:grassbp} is associated with a change in stability of the grass-only fixed point, and with a grass-tree coexistence fixed point entering or leaving $\mathbb{R}_{\geq 0} \times \mathbb{R}_{\geq 0}$.

\begin{theorem} \label{thm:forestbp}
A necessary condition for a transcritical bifurcation associated with a grass-tree coexistence fixed point crossing into the first quadrant at $(0,1)$, is met when  \begin{equation}1+e^{(1-\alpha)\tau}(-1+k_1)=0. \label{eqn:thmforestbp} \end{equation}
\end{theorem}

\begin{proof}
For all values of the disturbance parameters, $(0,1)$ is a forest-only fixed point of the system, i.e. 
there is a forest only fixed point when 
\begin{equation}
\left\{\begin{array}{l}
    0=U(0,1)\\
    1=V(0,1)\end{array}\right.
\end{equation}

A transcritical bifurcation involving the forest-only state occurs when the $(0,1)$ branch of fixed points crosses a branch of fixed points of the form $(x^*,y^*)$ with $x^*$ not uniformly zero. Near the transcritical bifurcation $x-U(x,y)=0$ has a $x=0$ root, corresponding to the $(0,1)$ branch, and a root with $x$ almost always non-zero, corresponding to the $(x^*,y^*)$ branch. The transcritical bifurcation occurs when these two branches cross so that the $x=0$ root of $x-U(x,1)=0$ has multiplicity two. We construct a function $g(x)$ that reduces the multiplicity of the $x=0$ solution.

Let $g(x)=\dfrac{x-U(x,1)}{x}=1-\dfrac{U(x,1)}{x}$=0. 
We give explicit formulas for $U(x,1)$ and $g(x)$
Setting $y = 1$, integrating equation \ref{eqn:nondim}, and incorporating the impulses in equations \ref{eqn:grass} and \ref{eqn:trees}, we find
\begin{equation}
    U(x,1)=\frac{(1-\alpha)e^{(1-\alpha)t}(1-k_1)x}{1-\alpha-x+e^{(1-\alpha)t}x}
\end{equation}
Then 
\begin{equation}
\label{eqn:gx}
g(x)=1-\frac{U(x,1)}{x}=1-\frac{(1-\alpha)e^{(1-\alpha)t}(1-k_1)}{1-\alpha+x\left(e^{(1-\alpha)t}-1\right)}
\end{equation}

Setting Equation \ref{eqn:gx} to zero, we find $x = 1-e^{(1-\alpha)t}(1-k_1)$ is a root of $g(x)=0$. This crosses $x=0$ when 
$1-e^{(1-\alpha)t}(1-k_1)=0$, giving a repeated root $x=0$ to the equation $x=U(x,1)$. Thus $1-e^{(1-\alpha)t}(1-k_1)=0$ is a necessary condition for a transcritical bifurcation.  \hfill$\square$
\end{proof}

Note that in Theorem \ref{thm:forestbp}, we show that the location of the bifurcation curve varies with the indicators of fire intensity for grass, $k_1$, and fire frequency, $\tau$, but does not depend on the indicator of fire intensity for trees, $k_2$.

The bifurcation condition in Theorem \ref{thm:forestbp} is shown as a blue surface in Figure \ref{fig:bifthms}A and as a blue curve in Figures \ref{fig:bifthms}BCD.  Crossing this condition is associated with a change in stability of the forest-only fixed point, and with a grass-tree coexistence fixed point entering or leaving $\mathbb{R}_{\geq 0} \times \mathbb{R}_{\geq 0}$.


\begin{figure}
    \centering
  \includegraphics{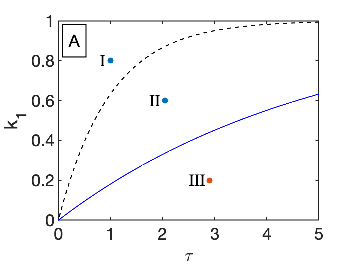}
    \includegraphics{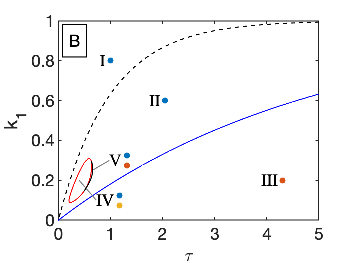}

  \includegraphics{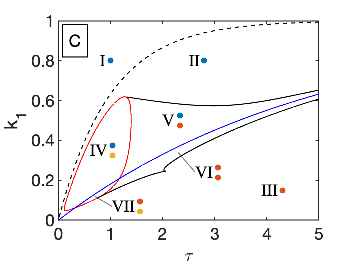}
    \includegraphics{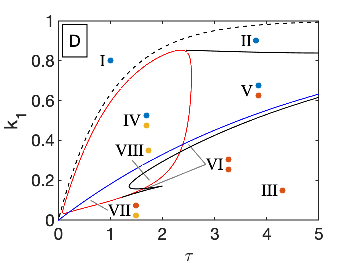}\\
    
    \begin{tabular}{|c|l|}
     \hline
     \tikz\draw[matlabcolor1,fill=matlabcolor1] (0,0) circle (.5ex); & stable $(0,1)$ forest-only fixed point \\ \hline
      \tikz\draw[matlabcolor3,fill=matlabcolor3] (0,0) circle (.5ex); & stable $(x_g,0)$ grass-only fixed point \\ \hline
     \tikz\draw[matlabcolor2,fill=matlabcolor2] (0,0) circle (.5ex); & stable coexistence fixed point 
     \\ \hline
\end{tabular}
    \caption{These stability diagrams show transcritical and limit point bifurcation locations, as well as the type of stable fixed point(s) in each regions, in the $\tau k_1$-parameter space for $k_2=0.2,0.4,0.6$, and $0.8$ in A,B,C,D respectively. See Table \ref{table:stability_regions} for a list of fixed points in each stability region.
    \textbf{(A)} $k_2=0.2$. The dashed black curve is a transcritical bifurcation of the origin (a zero biomass state) and a grass-only fixed point at $(x_g,0)$ as given in \eqref{eqn:grass-fp existence}. The blue curve is a transcritical bifurcation of the forest only state with a coexistence (savanna) state. At this curve, we see the stable forest from above has switched to a stable coexistence savanna state below. These two curves do not depend on $k_2$ and appear in B,C, and D as well. For any parameter set, there is one stable fixed point for $k_2=0.2$.
    \textbf{(B)} $k_2=0.4$ For $k_2>0.349$, the red surface in Figure \ref{fig:bifthms}A will occur as the red bifurcation curve. Along with a small limit point curve shown in black, this grass bifurcation gives two new stability regimes, both of which have bistability. 
    \textbf{(C)} $k_2=0.6$ As $k_2$ increases, the limit point and grass (red) bifurcation curves encompass a larger region of parameter space. Because they cross the blue (forest) bifurcation there are two additional stability regimes, which both have bistability. One includes two stable coexistence states. 
    (D) $k_2=0.8$ The limit point curve has changed shape adding one more stability case. Figure \ref{fig:PhasePortraits_for_6C} and Figure \ref{fig:PhasePortraits_for_6D} give examples of phase portraits for each stability region in figures C and D.}
    \label{fig:four tau k1 stab}
\end{figure}

\begin{table}
\centering
 \begin{tabular}{|c|p{1.5cm}|p{1.5cm}|p{4cm}|}
     \hline
     region  & tree-only fixed point & grass-only fixed point & coexistence fixed points \\
     \hline
     I & \tikz\draw[matlabcolor1,fill=matlabcolor1] (0,0) circle (.5ex); stable & \hspace{1.2ex} - & \hspace{1.2ex} -\\ \hline
     II &\tikz\draw[matlabcolor1,fill=matlabcolor1] (0,0) circle (.5ex);  stable &  \hspace{1.2ex} saddle & \hspace{1.2ex} -\\ \hline
     III & \hspace{1.2ex} saddle &  \hspace{1.2ex} saddle & \tikz\draw[matlabcolor2,fill=matlabcolor2] (0,0) circle (.5ex); stable \\ \hline
     \textbf{IV}  & \tikz\draw[matlabcolor1,fill=matlabcolor1] (0,0) circle (.5ex); \textbf{stable} & \tikz\draw[matlabcolor3,fill=matlabcolor3] (0,0) circle (.5ex); \textbf{stable} &   \hspace{1.2ex} saddle \\ \hline
     \textbf{V} &  \tikz\draw[matlabcolor1,fill=matlabcolor1] (0,0) circle (.5ex); \textbf{stable} &  \hspace{1.2ex} saddle &  \tikz\draw[matlabcolor2,fill=matlabcolor2] (0,0) circle (.5ex); \textbf{stable},  saddle \\ \hline 
     \textbf{VI} &  \hspace{1.2ex} saddle &  \hspace{1.2ex} saddle & \tikz\draw[matlabcolor2,fill=matlabcolor2] (0,0) circle (.5ex); \textbf{stable}, \tikz\draw[matlabcolor2,fill=matlabcolor2] (0,0) circle (.5ex); \textbf{stable},   saddle \\ \hline
     \textbf{VII}  & \hspace{1.2ex} saddle & \tikz\draw[matlabcolor3,fill=matlabcolor3] (0,0) circle (.5ex); \textbf{stable} & \tikz\draw[matlabcolor2,fill=matlabcolor2] (0,0) circle (.5ex); \textbf{stable},  saddle \\ \hline
     VIII  &  \hspace{1.2ex} saddle & \tikz\draw[matlabcolor3,fill=matlabcolor3] (0,0) circle (.5ex); stable & \hspace{1.2ex} - \\ \hline \end{tabular}
\caption{The nontrivial fixed points and their stabilities for all eight regions from Figure \ref{fig:four tau k1 stab} is shown here. There are three types of fixed points: tree-only (forest) at $(0,1)$, grass-only (grassland) at $(x_g,0)$ where $x_g$ is determined as in \eqref{eqn:grass-only fixed point}, and coexistence (savanna) fixed points with nontrivial values of both $x$ and $y$. The coexistence values can range greatly in $x$ and $y$, and are examined further in the discussion. Additionally, this table shows there are four different types of bistability in this model (regions and bistable points are in bold): forest-grassland (Region IV), forest-savanna (Region V), savanna-savanna (Region VI) and grassland-savanna (Region VII).}
\label{table:stability_regions}
\end{table}

\subsection{Stability Diagrams}\label{sec:stabdiagrams}

Examining the impact of different disturbance regimes on the long term behavior of the system involves varying $k_1$, $k_2$, $\tau$.  The possible long term behaviors of the system are summarized in the stability diagrams shown in Figure \ref{fig:four tau k1 stab} and in Table \ref{table:stability_regions}.  In the figure, each stability region is labeled with a numeral and colored dots showing which fixed points are stable. Table \ref{table:stability_regions} gives a full description of all fixed points (stable or unstable) in each stability region. A representative phase portrait for each region of $k_1\tau$-space for $k_2 = 0.6$ and of $k_1\tau$-space for $k_2 = 0.8$ is included in the Appendix, in Figures \ref{fig:PhasePortraits_for_6C} and
\ref{fig:PhasePortraits_for_6D} respectively.

In Figure \ref{fig:four tau k1 stab}, we show diagrams in $k_1\tau$-space for four values of $k_2$ where $k_2$ is the potential maximum impact of the fire on the tree population.  At low potential impacts ($k_2 = 0.2$), in Figure \ref{fig:four tau k1 stab}A, no bistability regions exist, meaning there is a unique long-term behavior at each $k_1\tau$ parameter set.  Region III has a tree-grass coexistence fixed point as the long term behavior, while other regions have a forest-only state as the long term behavior.

The existence of the grass-only bifurcation curve and limit point curves at higher values of $k_2$ correspond to bistable regions of parameter space, where there are multiple possible long term behaviors for the system, based on its initial state.  The size of the bistable region in $k_1\tau$-space increases as $k_2$ increases.  These bistability regions can involve many different pairings of alternative stable states. We see bistability of grass-only and forest-only states (Region IV), of grass-only and tree-grass coexistence states (Region VII), of forest-only and tree-grass coexistence (Region V), and of two distinct tree-grass coexistence states (Region VI).  Fire frequency can be a determinant of which states are bistable.  A forest state is bistable with a grass-only state with more frequent fires (smaller values of $\tau$), as seen in region IV.  It is bistable with a tree-grass coexistence (savanna) state with less frequent fires, as seen in region V.  In the context of this model, observing a forest state, one would need to determine the disturbance regime to know the other possible stable state.

\subsection{Derivation of analogous continuous system}
\label{sec:ctsdisturb}

We derive a continuous system in the limit of small frequent disturbances.  Many models of savanna ecosystems use this kind of simple disturbance regime, rather than using an impulsive disturbance model \cite{accatino2010tree,beckage2009vegetation,staver2011tree,touboul2018complex}.  By building a comparable continuous system, we are able to show that many of the bistability and coexistence cases that exist in our discrete model, are not present in its continuous analog.


We compare systems with discrete disturbance to a system with continuous disturbance by considering the average disturbance over a period of time, and holding this constant. Fix $r_1 = k_1/\tau$, a measure of fire intensity, in a fixed ratio as $\tau\to 0$.  This measure increases with increasing $k_1$, the proportion of grass biomass that burns.  It also increases with more frequent fires, corresponding to decreasing $\tau$. 

\begin{lemma}
For the impulsive model defined by equations \ref{eqn:nondim}, \ref{eqn:grass}, \ref{eqn:trees}, the continuous analog system generated by fixed fire intensities $r_1 = k_1/\tau$ and $r_2 = k_2/\tau$ is \begin{equation}
    \begin{split}
 \label{eqn:ctssystemresult}
     \frac{dx}{dt}&=x(1-r_1-x-\alpha y)\\
           \frac{dy}{dt}&=\delta y(1-y).
    \end{split}
\end{equation}
\end{lemma}
\begin{proof}Let $\Phi(x,y,k_1,k_2,\tau)$ be the flow kick map such that
\begin{equation}
\label{eqn:flowkickmap2}
\begin{bmatrix}x_{n+1}\\y_{n+1}\end{bmatrix}=\Phi(x_n,y_n)=\begin{bmatrix} U(x_n,y_n)\\ V(x_n,y_n)\end{bmatrix}.\end{equation}
To find the continuous system, we take the limit of the difference quotient as $\tau \to 0$
\begin{equation}
    \begin{split}
 \label{eqn:ctssystem}
    \frac{dx}{dt} = \lim\limits_{\tau\to0} \frac{U(x,y)-x}{\tau} \\
    \frac{dy}{dt} = \lim\limits_{\tau\to0} \frac{V(x,y)-y}{\tau}
    \end{split}
\end{equation}

$U$ and $V$ are the flow kick map associated with equation \ref{eqn:nondim} and the disturbances in equations \ref{eqn:grass} and \ref{eqn:trees}.  Their Taylor expansions are given by $U(x,y) = (x + \tau f(x,y) + \mathcal{O}(\tau^2))(1-k_1)$ and $V(x,y) = (y + \tau g(x,y) + \mathcal{O}(\tau^2))(1-k_2\omega(k_1\phi_{\tau}(x,y))$, with $f, g$ from equation \ref{eqn:nondim} and $\phi_{\tau}$ the flow associated with $f$.
We have 
\begin{align*}
    \lim\limits_{\tau\to0} \frac{U(x,y)-x}{\tau} &=\lim\limits_{\tau\to0}\frac{(x+\tau f(x,y)+\mathcal{O}(\tau^2))(1-r_1\tau)-x}{\tau} \\
    &=-r_1x+f(x,y) 
\end{align*}
and
\begin{align*}
    \lim\limits_{\tau\to0} \frac{V(x,y)-y}{\tau} &=\lim\limits_{\tau\rightarrow 0}\frac{(y+\tau g(x,y) + \mathcal{O}(\tau^2))\left(1-r_2\tau\omega(r_1\tau\phi_\tau(x,y))\right) - y}{\tau} \\
    &= g(x,y)-r_2 \omega\left(0\right)y \\
&= g(x,y)
\end{align*}
The analogous continuous disturbance system is thus
\begin{equation}
    \begin{split}
 \label{eqn:ctssystem2}
     \frac{dx}{dt}&=x(1-r_1-x-\alpha y)\\
           \frac{dy}{dt}&=\delta y(1-y).
    \end{split}
\end{equation}
\hfill$\square$
\end{proof}

The impact of fire manifests as a $-r_1x$ term in the grass dynamics and does not appear in the tree dynamics.  This is because the argument to $\omega(.)$ is $k_1x$ and we have assumed $k_1\rightarrow 0$ as $\tau \rightarrow 0$.

However, if we do not consider the analogous continuous disturbance, but instead look for a continuous system that is more structurally analogous to the flow-kick system, we can assume the fire impact is $r_2 \omega(x)$, where $r_2 = k_2/\tau$ is held in a fixed ratio, so that tree mortality is dependent on the amount of grass present.  In this case, we find
\begin{equation}
    \begin{split}
 \label{eqn:ctssystem3}
     \frac{dx}{dt}&=x(1-r_1-x-\alpha y)\\
           \frac{dy}{dt}&=\delta y(1-y)-r_2\omega(x)y.
    \end{split}
\end{equation}

In addition, we can assume the fire impact is $r_2 \omega(r_1x)$, where $r_2 = k_2/\tau, r1 = k_1/\tau$ are held in a fixed ratio, so that tree mortality is dependent on the amount of grass burned ($r_1x$). Then the continuous system is  \begin{equation}
    \begin{split}
 \label{eqn:ctssystem4}
     \frac{dx}{dt}&=x(1-r_1-x-\alpha y)\\
           \frac{dy}{dt}&=\delta y(1-y)-r_2\omega(r_1 x)y.
    \end{split}
\end{equation}
This system has a similar $\omega$ input to the flow-kick system (although it is not the continuous analog of that system).

The system in equation \ref{eqn:ctssystem2} is a one-sided Lotka-Volterra system with very similar dynamics to the original underlying system.  It does not show the new behaviors that are present in the impulsive system.  The bifurcation structure in the $r_1r_2$-plane is shown for equation \ref{eqn:ctssystem3} in figure \ref{fig:ctsbifn}A and is shown for equation \ref{eqn:ctssystem4} in figure \ref{fig:ctsbifn}B.  Comparing figure \ref{fig:ctsbifn}B to figure \ref{fig:bifthms}B, these bifurcation curves are analogous to the bifurcations seen in the $k_1k_2$-plane at a fixed value of $\tau$.  Note that the small difference between these systems ($\omega(x)$ vs $\omega(r_1x)$ in the fire disturbance) leads to large qualitative differences in the bifurcation structure.  These systems are discussed further in section \ref{sec:ctsdistrubdiscuss} of the discussion.

\begin{figure}
\centering
\includegraphics[]{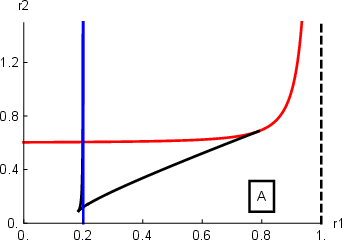} 
\includegraphics[]{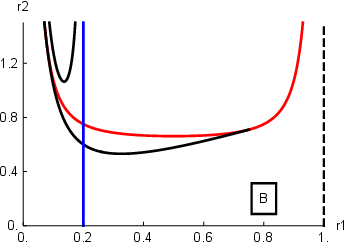}
\includegraphics[trim = .1in 1.5in .1in .55in, clip]{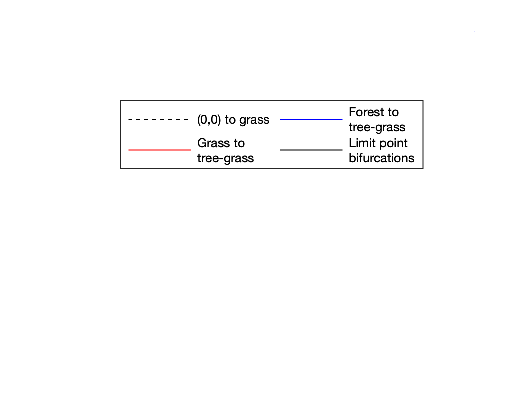}
\caption{ 
\textbf{(A)} shows the bifurcation structure for Equation 
\ref{eqn:ctssystem3}, where the continuous tree disturbance is given by $-r_2\omega(x)$, while \textbf{(B)} shows the bifurcation structure for Equation \ref{eqn:ctssystem4}, where the continuous tree disturbance is given by $-r_2\omega(r_1x)$.
The dashed black line indicates the transcritical bifurcation where a nontrivial grass-only fixed point comes into existence, the red curve is the transcritical bifurcation marking the grass-only to coexistence transition, and in blue the transcritical bifurcation marking the transition from stable forest to stable coexistence. Note that both the grass-only branch and forest to coexistence transcritical bifurcations are not dependent on $r_2$. The black curves correspond to saddle-node or fold bifurcations. These figures can be compared to Figure \ref{fig:bifthms}D, the $k_1k_2$-cross section, with fixed $\tau$, of the full impulsive system.}
\label{fig:ctsbifn}
\end{figure}

\section{Discussion} \label{sec:discussion}

We have a presented a complete bifurcation analysis in disturbance parameter space for a simple impulsive model of a savanna.  We find numerous regions of bistability, including different pairs of bistable states.  We also find bifurcation curves that are bounded away from the origin in disturbance space, meaning that they only occur in the impulsive version of the model.  The complexity of the stability diagram in figure \ref{fig:four tau k1 stab}D shows that substantial variation can occur in the long term behavior with small differences in disturbance parameters.

In this discussion, we compare the predictions of this model to those of related continuous models.  We also examine the resilience of the savanna ecosystem to different disturbance regimes.  Using a socially valued function of the system state, we assess the impact of small changes in the disturbance regime on the socially valued quantity.  We also consider whether this analysis, which is focused on long term behavior, extends to predictions relevant on shorter timescales.

 



\subsection{Comparison to continuous system} \label{sec:ctsdistrubdiscuss}

The continuous system in equation \ref{eqn:ctssystem2} is directly analogous (via a small fire interval limit) to the flow-kick system.  This continuous analog system has Lotka-Volterra type behavior and exhibits none of the bistability of the flow-kick system.  However, by trying other versions of the continuous fire disturbance that include a Hill function, as in equations \ref{eqn:ctssystem3} and \ref{eqn:ctssystem4}, we are able to find bifurcations similar to those seen in the $k_1k_2$-plane (for fixed $\tau$) in the flow-kick system.  These continuous systems do not have a direct mathematical relationship to the impulsive system that we investigated. Note that equation \ref{eqn:ctssystem4} is closest to those in the literature. For example \cite{staver2011tree,touboul2018complex} both make use of a Hill function with argument of only the amount of grass present in the system, similar to $\omega(x)$. (This is used to temper the growth of tree saplings to adult trees, based on the assumption of greater mortality due to fire of saplings compared to trees.)

The similarities between the bifurcation structure of equation \ref{eqn:ctssystem4} (see figure \ref{fig:ctsbifn}) and the flow-kick system suggest that the dynamical roles of $r_1$ and $r_2$ in these continuous systems are similar to those of $k_1$ and $k_2$.  However, $r_1$ and $r_2$ ostensibly combine fire intensity and frequency in a single variable.  Based on this comparison of bifurcation structures, capturing the impact of fire frequency seems to require using the full impulsive system, rather than a continuous variant.

\begin{figure}
  \includegraphics{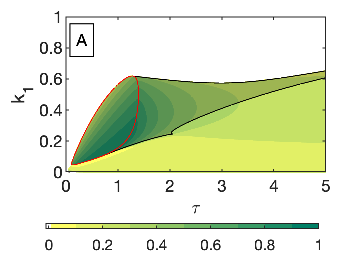}
 \includegraphics[]{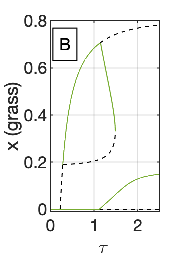}
 \includegraphics[]{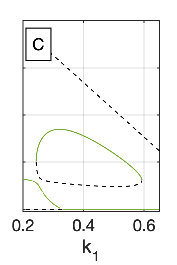}
\caption{\textbf{(A)} The amount of grass, $x$, associated with a stable fixed point of the system is indicated by color (see legend). The region of the space with bistable fixed points is shaded slightly grey and surrounded by the black lines and red curve. In this bistable region, there are two potential grass values, and the color corresponds to the higher grass value. 
\textbf{(B)} This is a bifurcation diagram in $\tau$ for $k_1=0.2$ and $k_2=0.6$. The values of $x$ (grass) associated with stable fixed points are shown in green and associated with saddle points in dashed black.  The sharp transition between high and low grass values that is visible along the black limit point curve in (A) corresponds to the limit point bifurcation at $\tau=1.49$ in (B). Crossing this transition by increasing $\tau$, starting from a high grass state, (for fixed $k_1$ and $k_2$) would be irreversible, with $x$ remaining in the low grass state when $\tau$ returned to a value below the limit point.
\textbf{(C)} This is a bifurcation diagram showing $x$ and $k_1$ for $\tau=2$ and $k_2=0.6$.  Similarly to (B), a green curve denotes the $x$ value for stable fixed points and black denotes it for saddle points.  Also similarly to (B), adjusting $k_1$ (the fire mortality of grass) in either direction with $\tau$ and $k_2$ fixed, starting from the high grass state, would result in a transition to a low grass state.  That transition would not reverse if the parameter value was restored to its original value.
}
\label{fig:discussionfig1}       
\end{figure}

\begin{figure}
  \includegraphics{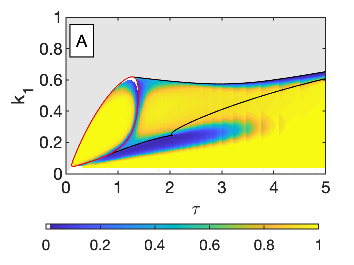}
  \includegraphics{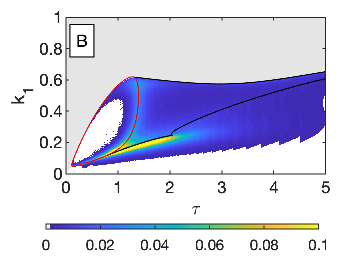}

\caption{
We examine the difference between long-term behavior and shorter-term behavior by comparing the equivalent of $30$ years of flow-kick iterations to the limiting behavior of the system as time goes to infinity.  \textbf{(A)} For each disturbance parameter set ($k_2 = 0.6$, $k_1, \tau$ vary), we consider all initial conditions within the basin of attraction for the long term state shown in \ref{fig:discussionfig1}A.  We find the proportion of initial conditions that have a grass value after $30$ years that is within $5$\% of the long-term value.  The colormap indicates this proportion.  \textbf{(B)} Using the same set of initial conditions as in (A), we calculate the distance between the $30$ year grass value and the long-term value.  The colormap is indicating the mean value of these distances. 
}
\label{fig:discussionfig2}       
\end{figure}

\subsection{Resilience analysis} \label{sec:resanalysis}

One motivation for studying a savanna ecosystem via a model is to identify how vulnerable the system would be to changes, such as increased fire intensity or increased fire frequency, that may be associated with climate change or other forcing. We use the term resilience to refer to the effect of a change in disturbance on some valued aspect of the system.  There are a number of ways to assign value to an ecosystem.  Some of these valuations are functions of the system state; we call such functions socially valued functions.  To examine resilience, we study how a socially valued function of the system state will vary with changes in fire disturbance parameters.  In our example below, we will show that small changes in the fire intensity or fire frequency lead to irreversible changes in the amount of grass, $x$, present in the system, meaning that this socially valued function is not resilient to those small changes in disturbance.



The amount of grass, $x$, may be of interest when a savanna is being valued for its grazing potential.  We choose to use the long-term value of $x$ just before the disturbance is applied as a measurement of the amount of grass.  Other choices, such as time-averaging over a cycle, or using the amount of grass just after the disturbance is applied, would yield a qualitatively similar analysis.  Rather than assigning a social value, $f(x)$, we set the social value to $x$ itself for this example.  At disturbance parameter sets with bistability, we choose the highest grass state available as the social value.  The socially valued state is shown in Figure 
\ref{fig:discussionfig1}A.

Of note in this figure is the stark transition from low to high social value along the curve of saddle-node bifurcations to the left of the cusp point. In \ref{fig:discussionfig1}B and C, we provide bifurcation diagrams that illustrate a sharp transition in the social value that is not reversible.  For example, starting at a fire interval of $\tau = 1$ and increasing the fire interval to $\tau = 2$, leads to a marked decrease in $x$.  Returning to $\tau = 1$, the social value is trapped on the lower curve, where $x = 0$, so value is not restored. This irreversibility is more extreme than hysteresis, as the valued state of the system is simply not accessible via adjusting the fire interval.

For this example of a socially valued function (the amount of grass), the bistability region defines a disturbance parameter space within which the socially valued function exhibits resilience to parameter variation.  The system state is not resilient to a shift that leads to a disturbance regime outside of bistability region, in that restoring the old disturbance regime does not restore the system state. Under this model, woody encroachment due to a change in disturbance would require more than just restoring the original fire regime. It would also require actively removing trees and promoting grass to return the system to the basin of attraction of the high grass state.

\subsection{Analyzing the system on decadal timescales} \label{sec:timescale}
Our previous analysis in this paper has relied upon identifying the long term behavior of the system.  However, we must take into account the possibility of transient timescales \cite{morozov2020long} as the timescales of interest for savanna-management are not infinite.  Here we consider a 30-year timescale to examine the relationship between medium term transient behavior (the 30-year behavior) and the long-term behavior of the system. We see that for $k_2=0.6$, in the majority of $\tau k_1$-parameter space, the 30 year transient behavior varies minimally from the long term behavior, with the exception of parameter sets near the bifurcation curves where dynamics are naturally slower.

In Figure \ref{fig:discussionfig2}, there are two visualizations comparing the highest pre-kick grass value in a periodic equilibrium of the impulsive system with the highest grass value after a decadal time scale close to 30 years. Thus we can conclude that the long term behavior analysis (which is easier to compute) gives a good proxy for a shorter time transient analysis in this model.

\subsection{Contextualization of results} \label{sec:context}

In this section we compare the results of our model to other models cited in this paper.

In field observations, savanna systems exhibit bistability between attracting states and have an attracting tree-grass coexistence state.  \cite{staver2011tree,staver2011global}. In modeling efforts from the early 2000s that were working to capture these observations, Beckage et al \cite{beckage2009vegetation} conclude that a grass-fire feedback in a two- or three-species ODE model does not lead to a stable savanna state where trees and grass coexist.  Rather, using a savanna tree-fire feedback stabilizes a savanna state in their three-species system.  Accatino et al \cite{accatino2010tree} incorporates soil moisture, along with a fire term, to stabilize a savanna state in a three dimensional model with two plant species and soil moisture.  In this work, we find that a two species tree-grass interaction model has stable savanna states with the introduction of impulsive disturbances.

In addition finding mechanisms that generate a stable savanna state, savanna modeling efforts also focus on capturing bistabilities in the system.  For example, Staver et al \cite{staver2011tree} use observations to calculate aerial tree cover in proportion to aerial grass cover in savanna regions.  They find that savannas and forests exist as alternative stable states and use a three-species ODE model to capture this result.  Other papers generate similar bistability results using stochastic fire feedbacks \cite{batllori2015minimal,baudena2010idealized}.  Some stochasticity, e.g. \cite{patterson2020probabilistic}, is also captured by continuum models like those in \cite{staver2011tree,touboul2018complex}.
The impulsive model presented in this paper not only has forest-savanna bistability, but forest-grassland bistability and bistability between woodier and grassier savanna co-existence states. Similar modeling work that uses a different parameterization of the impact of fires on trees, so as to capture age-effects in a two-species model \cite{yatat2021minimalistic,tamen2017minimalistic}, shows it is possible to generate a similar diversity of bistabilities with a somewhat different disturbance regime.

Spatially homogenized models, either continuous or stochastic, capture a mean-field average of trees and grass over the whole area of a savanna. Some, like the model in this paper, focus on biomass, while others model the fraction of coverage for different plants.  In contrast, spatially extended models have the potential to capture patterns of trees and grass within the savanna, as well as the boundaries between savanna and forest \cite{wuyts2019tropical,yatat2018spatially}. For example, in 
\cite{goel2020dispersal}, Goel et al create a reaction-diffusion model which reflects natural biome boundaries between savanna and forest, and conclude that biome recovery may be easier in a spatially extended model than it is in a nonspatial model. This is in stark contrast to the stronger-than-hysteresis tipping that we show exists in the impulsive model explored in this paper (see Figure \ref{fig:discussionfig1}). 


\subsection{Future Work} \label{sec:future}

This analysis of a simple one-sided Lotka-Volterra interaction with an impulsive disturbance focused on bifurcations in the disturbance space.  Understanding sensitivity of the predictions of the model to the structure of the underlying Lotka-Volterra interaction, including to parameters in that system, would extend these results to a broader range of savanna systems.  In addition, it would make it possible to separately assess the robustness of model predictions in the face of structural model errors, which are inevitable in ecological modeling.

It would be possible to focus further work using this model on questions of woody encroachment, as transitions between grassy and woody states are observed in the model. Creating a spatially extended impulsive model could also be of interest, particularly in light of the ease of biome recovery in spatial models. In addition, the analysis of transient behaviors could be placed within a more general framework.  Given that many analyses focus on long term behavior, expanding tools for understanding transient would be worthwhile.

\begin{acknowledgements}
Conversations with Mary Lou Zeeman and with Katherine Meyer provided valuable insights.  Part of this work was completed during A.H-L.'s sabbatical, which was partially supported by the Hutchcroft Fund and the Mathematics and Statistics Department at Mount Holyoke College.
\end{acknowledgements}
\bibliographystyle{spmpsci}
\bibliography{savanna}
\appendix
\section{Appendix: Phase Portraits}
\label{sec:appendix}

We construct a representative phase portrait for each region of figures \ref{fig:four tau k1 stab}C and \ref{fig:four tau k1 stab}D.  We use $\alpha = 0.8, \delta = 0.6$, $a = 0.08$, and indicate the values of $k_1, k_2,\tau$ in the title of each phase portrait.  Phase portraits for $k_2 = 0.6$ are shown in Figure \ref{fig:PhasePortraits_for_6C} and phase portraits for $k_2 = 0.8$ are shown in Figure \ref{fig:PhasePortraits_for_6D}.  Fixed points are indicated by filled circles.  Red circles are used for repelling fixed points, gray for saddle fixed points, and shades of green are used for stable fixed points.  The entire periodic equilibrium is shown as a curve attached to the flow-kick map fixed point.

Saddle fixed points have associated stable and unstable manifolds.  The stable manifold forms a basin boundary between different basins of attraction in the phase portrait.  Orbits of the flow-kick map quickly collapse to the unstable manifold, so most evolution of the system is along this manifold.  The flow region associated with the manifold is the shaded region in the center of each phase portrait.

\begin{figure}
    \hspace{-0.4in}\includegraphics[]{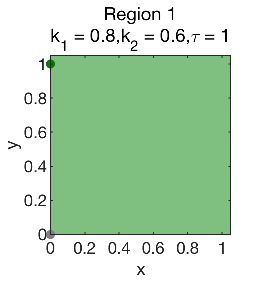}
    \includegraphics[]{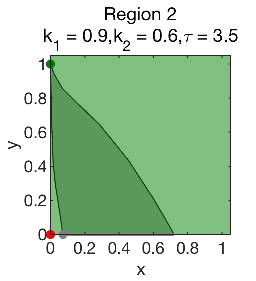}
    \includegraphics[]{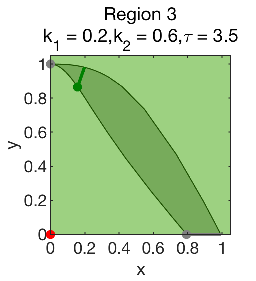}
    
    \hspace{-0.4in}\includegraphics[]{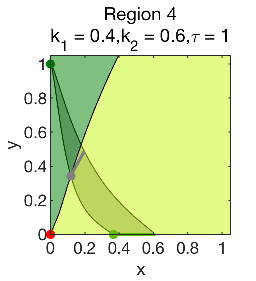}
    \includegraphics[]{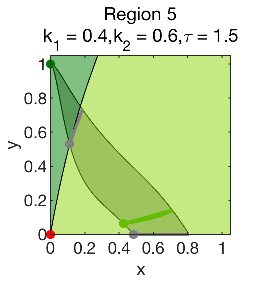}
    \includegraphics[]{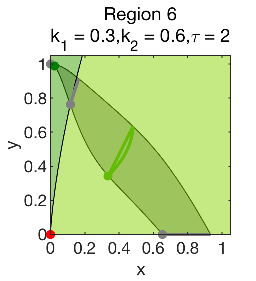}
    
    \hspace{-0.4in}\includegraphics[]{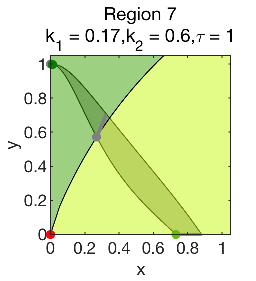}
\caption{Phase Portraits for stability regions with $k_2=0.6$.  The figures are described in the text of Appendix \ref{sec:appendix}}
\label{fig:PhasePortraits_for_6C}    
\end{figure}

\begin{figure}
   \hspace{-0.5in} \includegraphics{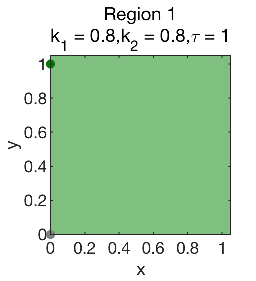}
    \includegraphics[]{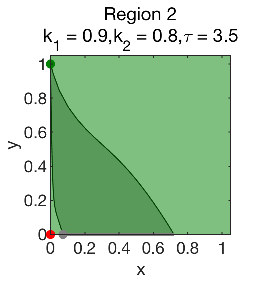}
    \includegraphics[]{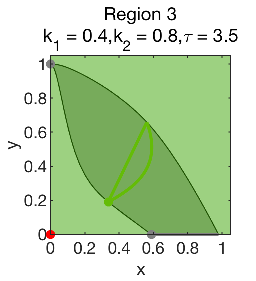}
    
   \hspace{-0.5in} \includegraphics[]{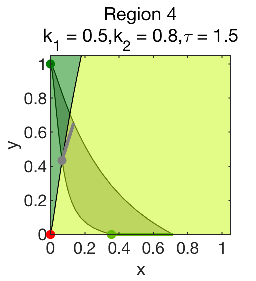}
    \includegraphics[]{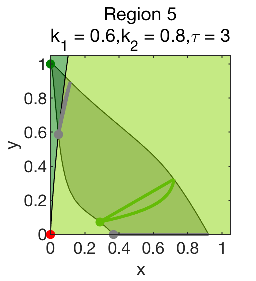}
    \includegraphics[]{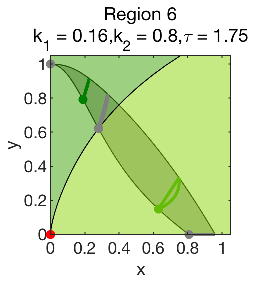}
    
    \hspace{-0.5in}
    \includegraphics[]{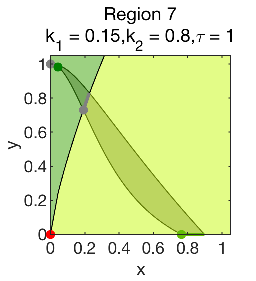}
    \includegraphics[]{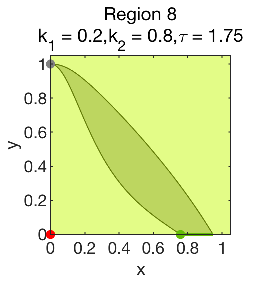}
    \caption{Phase Portraits for stability regions with $k_2=0.8$.  The figures are described in the text of Appendix \ref{sec:appendix}.}
    \label{fig:PhasePortraits_for_6D}
\end{figure}

%
%


%
%


\end{document}